\documentclass[12pt]{article}
\usepackage{amsfonts}
\usepackage{amssymb}
\usepackage{graphicx}
\usepackage{amsmath}
\usepackage{harvard}
\usepackage{float}

\newtheorem{theorem}{Theorem}

\newtheorem{corollary}{Corollary}

\newtheorem{proposition}{Proposition}

\newenvironment{proof}[1][Proof]{\textbf{#1.} }{\ \rule{0.5em}{0.5em}}

\renewcommand{\cite}{\citeasnoun}
\begin{document}

\author{Dirk Bergemann\thanks{%
Yale University, dirk.bergemann@yale.edu.} \and Tibor Heumann\thanks{%
Pontificia Universidad Cat\'olica de Chile, tibor.heumann@uc.cl.} \and %
Stephen Morris\thanks{%
Massachusetts Institute of Technology, semorris@mit.edu}}
\title{The Optimality of Constant Mark-Up Pricing\thanks{%
We acknowledge financial support from NSF grants SES-2001208 and
SES-2049744, and ANID Fondecyt Iniciacion 11200365. We have benefitted from
comments by Yang Cai, Jason Hartline, Yingkai Li, and Zvika Neeman. We thank
Hongcheng Li for excellent research assistance.} }
\date{\today }
\maketitle

\begin{abstract}
We consider a nonlinear pricing environment with private information. We
provide profit guarantees (and associated mechanisms) that the seller can
achieve across all possible distributions of willingness to pay of the
buyers. With a constant elasticity cost function, constant markup pricing
provides the optimal revenue guarantee across all possible distributions of
willingness to pay and the lower bound is attained under a Pareto
distribution. We characterize how profits and consumer surplus vary with the
distribution of values and show that Pareto distributions are extremal. We
also provide a revenue guarantee for general cost functions. We establish
equivalent results for optimal procurement policies that support maximal
surplus guarantees for the buyer given all possible cost distributions of
the sellers.

\medskip

\noindent \textsc{Jel Classification: }D44, D47, D83, D84.

\newpage
\end{abstract}


\newpage

\section{Introduction\label{sec:into}}

\subsection{Motivation\label{subsec:mot}}

The arrival of digital commerce has lead to an increasing use of
personalization and differentiation strategies. With differentiated products
along the quality dimension and/or the quantity dimension comes the need for
nonlinear pricing policies or second degree price discrimination. The
optimal pricing strategies for quality or quantity differentiated products
were first investigated by \cite{muro78} and \cite{mari84}, respectively.
The optimal pricing strategies were shown to depend heavily on the prior
distribution of the private information regarding the types, and ultimately
the willingness-to-pay of the buyers.

Yet, in many circumstances the seller has only very weak and partial
information about the distribution of the types. The aim of this paper is to
devise robust pricing policies that: (i) are independent of the specific
distribution of the types and (ii) exhibit revenue guarantees across all
possible distributions. The main results of this paper construct such robust
pricing strategies and show that these pricing strategies can be expressed
in term of very simple and transparent pricing policies. We establish
revenue guarantees by comparing the profit under the robust pricing policy
and any given distribution of private information to the social welfare
attainable under the same distribution of private information. We then seek
to identify the pricing policy that guarantees the highest ratio of these
two measures across all feasible distributions with finite expectation. As
the social welfare coincides with the profit that the seller could attain
under perfect or first degree price discrimination, the ratio has two
possible interpretations. In this second perspective we compare the revenue
of the seller under incomplete information to the revenue that the seller
could attain under complete information (and hence perfect price
discrimination).

We consider two broad classes of pricing problems. First, we consider a
class of quality differentiated pricing problems as first analyzed by \cite%
{muro78}. Here, the type of the buyer is the willingness-to-pay for a
quality. The cost of production is an increasing and convex function of
quality. We focus on iso-elastic cost function which allows to express the
cost environment of the seller in terms of a constant cost elasticity.
Second, we consider a class of quantity differentiated pricing as first
analyzed by \cite{mari84}. Here, the seller has a constant marginal cost of
producing additional units of his good and the buyer has a concave utility
for quantities. In this environment, the elasticity of demand is initially
assumed to be constant across all types of the buyer.

For the environment of quality differentiated products we exhibit a pricing
mechanism that attains a profit guarantee that delivers at least a fraction
of the social surplus across all possible prior distributions of private
information. The profit guarantee that is expressed in the ratio is an
expression that only depends on the cost elasticity $\eta $\ of the quality.
The ratio is described by a polynomial only in terms of the cost elasticity
(Theorem \ref{th:1}). Interestingly, the optimal mechanism can be
implemented by a pricing policy that maintains a constant mark-up for each
additional unit. The mark-up is also expressed in terms of the elasticity
and is given by $\left( \eta -1\right) /\eta $ (Corollary\ \ref{cor:opt}).
For cost functions with an elasticity $\eta $ near $1$, the ratio is the
largest and is given by $1/e$ in the limit $\eta \rightarrow 1$. As the cost
function becomes more convex, the ratio decreases. With a quadratic cost
function, i.e., when $\eta =2$, the ratio is still $1/4$. Eventually as the
convexity of the cost becomes more pronounced, the ratio converges to $0$ as 
$\eta \rightarrow \infty $. We show that these profit guarantees are sharp
and are attained by specific instances of distributions, namely Pareto
distributions whose shape $\eta /\left( \eta -1\right) $\ varies
systematically with the elasticity of the cost function.

We then extend the scope of our analysis and give a complete
characterization of the efficient frontier of producer surplus and consumer
surplus that can be attained by optimal mechanisms across all possible
distributions and all possible cost elasticities. Here we express both the
ratios of the producer and consumer surplus relative to the social surplus.
In fact, for every elasticity, the largest share of consumer surplus is
attained by the very distribution under which the robust mechanism is the
optimal revenue maximizing mechanism. While the efficient frontier is of
immediate interest, we in fact obtain the entire set of pairs of attainable
surplus across all distributions. Interestingly, a version of the Pareto
distribution appears in all of the boundaries that form the set of
attainable surplus pairs. The analysis delivers an exact upper bound on the
profit guarantee when the cost function is entirely characterized by a
constant elasticity $\eta $. Yet, the mechanism that generates the bound,
namely the constant mark-up pricing rule also generates a lower bound on the
profit guarantee for arbitrary convex cost functions as long as the marginal
cost is convex as well (Proposition \ref{prop:low}).

We then turn to the model of quantity differentiation. Here the seller has a
constant marginal cost to provide additional units of his product. By
contrast, the buyer has a concave utility for the product and thus
diminishing marginal utility for additional units of the product. The demand
elasticity $\eta $ now summarizes the economics of the environment. The
arguments developed in the setting with nonlinear cost can be largely
transferred and yield equally sharp results.\ For every demand elasticity,
we obtain a profit guarantee that is polynomial of the demand elasticity
alone. Surprisingly, the robust mechanism is given by a linear pricing
mechanism

\subsection{Related Literature\label{subsec:rel}}

We derive performance guarantees and robust pricing policies that secure
these guarantees for large classes of second-degree price discrimination
problems. We do this without imposing any restriction on the distribution of
the values, such as regularity or monotonicity assumptions, or finite
support conditions. We only require that the social surplus of the
allocation problem has finite expectation.

The optimal monopoly pricing problem for a single object is a special
limiting case of our model when the marginal cost of increasing supply
becomes infinitely large. The analysis of a performance guarantee in the
single-unit case is also a special case of a result of \cite{neem03}. He
investigates the performance of English auctions with and without reserve
prices. The case of the optimal monopoly pricing is a special case of the
auction environment with a single bidder. He establishes a tight bound for
the single-bidder case that is given by a \textquotedblleft truncated Pareto
distribution\textquotedblright\ (Theorem 5). The bound that he derives is a
function of a parameter that is given by the ratio of the bidder's expected
value and the bidder's maximal value. Without the introduction of this
parameter the bound is zero: as this ratio converges to zero, so does the
bound. Similarly, \cite{haro14} establish that for distribution with support 
$\left[ 1,h\right] $, the competitive ratio $1+\ln h$. \ Thus as $h$ grows,
the approximation becomes arbitrarily weak (Theorem 2.1). By contrast, our
results obtain a constant approximation for every convex cost function.
Thus, the introduction of a convex cost function (or a concave utility
function) leads to noticeable and perhaps surprising strengthening of the
approximation quality.

\cite{carr17} considers a robust version of multi-item pricing problem. The
buyer has an additive value for finitely many items and has private
information about the value of each item. There the seller knows the
marginal distribution for every item but is uncertain about the joint
distribution of the valuation. \cite{carr17} considers a minmax problem and
shows that separate item-by-item pricing is the robustly optimal pricing
policy. \cite{dero22} consider a related problem under informational
robustness. There, the joint distribution of values is given by commonly
known distribution but nature or the buyer can choose the optimal
information structure. In the corresponding solution of the mechanism design
problem, they show that the optimal mechanism is always one of pure bundling.

The focus of this paper is on second-degree price discrimination alone. \cite%
{bebm15} consider the limits of third-degree price discrimination. While
most of their analysis is focused on the single-unit demand, they present
some partial results how market segmentation can affect second-degree price
discrimination. \cite{hasi22} present some additional results on the limits
of multi-product discrimination for a given prior distribution of values.

\section{Model\label{sec:model}}

A seller supplies goods of varying quality $q\in \mathbb{R}_{+}$ to a buyer.
The buyer has a willingness-to-pay (or value) ${v}\in \mathbb{R}_{+}$ for
quality. The utility net of the payment $t\in \mathbb{R}_{+}$ is: 
\begin{equation}
u({v},q,t)={v}q-t.  \label{u}
\end{equation}%
The value ${v}\in \mathbb{R}_{+}$ is distributed according to $F\in \Delta ([%
\underline{v},\bar{v}]),\,$with support $0\leq \underline{v}<\bar{v}\leq
\infty $. The seller's cost of providing quality $q$ is 
\begin{equation}
c(q)=q^{\eta }/\eta .  \label{ce}
\end{equation}
Note that the cost has constant elasticity equal to $\eta \in \left(
1,\infty \right) $.

An important benchmark is the social surplus generated by the efficient
allocation: 
\begin{equation}
S_{F}\triangleq \mathbb{E}[\max_{q}\{vq-c(q)\}]  \label{eq:ts}
\end{equation}%
This is clearly an upper bound on the profit and consumer surplus generated
by any mechanism. We assume that: 
\begin{equation*}
\lim_{v\rightarrow \infty }(1-F(v))v^{\frac{\eta }{\eta -1}}=0.
\end{equation*}%
This is a necessary and sufficient condition to guarantee that $S_{F}<\infty 
$.

The seller chooses a menu (or direct mechanism)\ with qualities $Q(v)$ at
prices $T(v):$%
\begin{equation}
M\triangleq \left\{ (Q(v),T(v))\right\} _{v\in \mathbb{R}}.
\end{equation}%
We use capital letters to denote functions and lower case letters to denote
real numbers. The mechanism has to satisfy incentive compatibility and
participation constraints. Thus for all $v,v^{\prime }\in \mathbb{R}:$ 
\begin{align}
vQ(v)-T(v)& \geq vQ(v^{\prime })-T(v^{\prime });  \label{eq:ic} \\
vQ(v)-T(v)& \geq 0.  \label{eq:ir}
\end{align}%
The expected seller's profit and buyer's surplus generated by a distribution 
$F$ and a mechanism $M$\ are respectively: 
\begin{equation*}
\Pi _{F,M}\triangleq \mathbb{E}[T(v)-c(Q(v))]\quad \text{and}\quad
U_{F,M}\triangleq \mathbb{E}[Q(v)v-T(v)].
\end{equation*}%
The optimal mechanism for distribution $F$ is: 
\begin{equation*}
M^{F}=\underset{M}{\arg \max }\ \Pi _{F,M}.
\end{equation*}%
With some abuse of notation, we denote by 
\begin{equation*}
\Pi _{F}=\Pi ^{F,M^{F}}\text{ and }U_{F}=U^{F,M^{F}}
\end{equation*}%
the profit and consumer surplus evaluated at the seller-optimal mechanism.

\section{A Profit-Guarantee Mechanism\label{sec:mini}}

\subsection{A Minmax Problem}

The first objective is to find the optimal profit-guarantee mechanism $M_{G}$
defined as: 
\begin{equation}
M_{G}=\underset{M}{\arg \max }\ \inf_{F}\ \frac{\Pi _{F,M}}{S_{F}}.  \tag{A}
\label{eq:revg}
\end{equation}%
As a direct by-product we find the distribution of values that minimizes the
seller's normalized profit: 
\begin{equation}
\inf_{F}\max_{M}\ \frac{\Pi _{F,M}}{S_{F}}.  \tag{B}  \label{eq:minprof}
\end{equation}%
In fact, we will show that the Minimax Theorem applies in our model, so %
\eqref{eq:revg} and \eqref{eq:minprof} will be tightly connected.

As in related work on optimal pricing, the Pareto distribution and truncated
versions of the Pareto distribution will play an important role. Towards
this end, we define the truncated version of the Pareto distribution by: 
\begin{equation}
P_{\alpha ,k}(v)\triangleq 
\begin{cases}
1-\frac{1}{v^{\alpha }} & \text{if }v<k; \\ 
1 & \text{if }v\geq k.%
\end{cases}
\label{trunc}
\end{equation}%
That is, $P_{\alpha , k}(v)$ is the same as a Pareto distribution for values 
$v<k$ and it has a mass point at $k$. We denote by 
\begin{equation*}
\lim_{k\rightarrow \infty }P_{\alpha ,k}=P_{\alpha }.
\end{equation*}%
In other words, when we omit the parameter $k$ from the truncated Pareto
distribution, it means we are taking the limit as $k\rightarrow \infty $. If 
$\alpha >1,$ then the expectation of the willingness-to-pay under the Pareto
distribution is finite, and if $\alpha >\eta /(\eta -1)$, then the social
surplus remains finite.

A mechanism $M$\ is incentive compatible and meets the individual
rationality constraint at $v=0$ if and only if the qualities $\{Q(v)\}_{v\in 
\mathbb{R}}$ are non-decreasing and the corresponding transfers are given
by: 
\begin{equation*}
T(v)=vQ(v)-\int_{0}^{v}Q(s)ds.
\end{equation*}%
Since the optimal mechanism is uniquely defined by the qualities, we often
refer to a mechanism as a set of qualities $\{Q(v)\}_{v\in \mathbb{R}}\ $%
omitting the prices. The profits generated by a mechanism are: 
\begin{equation}
\Pi \triangleq \int_{0}^{\infty }(vQ(v)-c(Q(v)))f(v)dv-\int_{0}^{\infty
}Q(v)(1-F(v))dv.  \label{profits}
\end{equation}%
We then have that the profit can be computed as a function of the qualities $%
\{Q(v)\}_{v\in \mathbb{R}}$ alone.

Under the Pareto distribution, the gross virtual values are given by: 
\begin{equation}
\phi (v)\triangleq \frac{\alpha -1}{\alpha }v,  \label{p1}
\end{equation}%
and so the qualities $Q\left( v\right) $\ supplied by the seller under the
Bayes optimal mechanism are given by: 
\begin{equation}
Q\left( v\right) =(\frac{\alpha -1}{\alpha }v)^{\frac{1}{\eta -1}}.
\label{p2}
\end{equation}

We now give the first main result, which provides the optimal
profit-guarantee mechanism.

\begin{theorem}[Profit Guarantee Mechanism]
\label{th:1}\quad \newline
The profit-guarantee mechanism $M^{\ast }$ is given by: 
\begin{equation}
Q(v)=\frac{v^{\frac{1}{\eta -1}}}{\eta ^{1/(\eta -1)}},  \label{q}
\end{equation}%
and generates normalized profit 
\begin{equation}
\frac{\Pi _{M^{\ast }}}{S_{F}}=\frac{1}{\eta ^{\frac{\eta }{\eta -1}}},
\label{q1}
\end{equation}%
for every $F$.%
\end{theorem}

\begin{proof}
We first prove that, for {every} distribution $F$, the profits generated by
qualities \eqref{q} are: 
\begin{equation*}
\Pi _{M}=\frac{1}{\eta ^{\frac{\eta }{\eta -1}}}S_{F}.
\end{equation*}%
Note that we write $\Pi _{M}$ because these are the profits generated by %
\eqref{q}, which in general differs from the optimal mechanism under $F$.
Replacing \eqref{q} into \eqref{profits}, the qualities, we get: 
\begin{equation*}
\Pi =\int_{0}^{\infty }(z-\frac{z^{\eta }}{\eta })v^{\frac{\eta }{\eta -1}%
}f(v)dv-\int_{0}^{\infty }zv^{\frac{1}{\eta -1}}(1-F(v))dv,
\end{equation*}%
where $z={1}/{\eta ^{1/(\eta -1)}}$. Integrating by parts the second term we
get: 
\begin{equation*}
\Pi =\int_{0}^{\infty }(z-\frac{z^{\eta }}{\eta })v^{\frac{\eta }{\eta -1}%
}dv-\frac{\eta -1}{\eta }\int_{0}^{\infty }zv^{\frac{\eta }{\eta -1}}f(v)dv.
\end{equation*}%
Collecting terms, we get: 
\begin{equation*}
\Pi =\int_{0}^{\infty }(z-\frac{z^{\eta }}{\eta }-\frac{\eta -1}{\eta }z)v^{%
\frac{\eta }{\eta -1}}f(v)dv.
\end{equation*}%
%
%
%
%
%
%
%
%
%
%
%
%
%
%
%
Replacing $z={1}/{\eta ^{1/(\eta -1)}}$ we get that: 
\begin{equation*}
\left. \Pi \right\vert _{_{z=\frac{1}{\eta ^{\frac{1}{\eta -1}}}}}=\frac{1}{%
\eta ^{\frac{\eta }{\eta -1}}}\frac{\eta -1}{\eta }\int_{0}^{\infty }v^{%
\frac{\eta }{\eta -1}}f(v)dv.
\end{equation*}%
We also note that the social surplus $S_{F}$ under distribution $F$ is: 
\begin{equation*}
S_{F}=\frac{\eta -1}{\eta }\int_{0}^{\infty }v^{\frac{\eta }{\eta -1}}f(v)dv.
\end{equation*}%
We thus get that this strategy guarantees a fraction $1/\eta ^{\frac{\eta }{%
\eta -1}}$ of the total surplus. Since \eqref{q} attains a fraction ${1}/{%
\eta ^{{\eta }/{\eta -1}}}$ of the social surplus, the infimum cannot be
smaller than this. But this fraction of the social surplus is exactly
attained by the Pareto distribution with shape $\alpha =\eta /(\eta -1)$.
Under this Pareto distribution, the above mechanism is the optimal
mechanism. Hence, this gives (both) the optimal profit-guarantee mechanism
and the minimum fraction of social surplus that can be generated by a
Bayesian-optimal mechanism for any distribution of values.
\end{proof}

This result establishes a mechanism that generates a profit guarantee at the
level given by (\ref{q1}). Interestingly this mechanism generates the \emph{%
same} normalized profit for every distribution $F$\ of values!

We now establish that the profit guarantee given by (\ref{q1}) is indeed the
optimal--the highest-- profit-guarantee that can provided. For this, we
first give the explicit solution to the Bayes-optimal mechanism \eqref{p1}
when the Pareto distribution has the shape parameter $\alpha =\eta /(\eta
-1) $ and the Bayes optimal mechanism equals the profit-guarantee
established in Theorem \ref{th:1}.

\begin{theorem}[Minimax Distribution]
\label{prop:2}\quad \newline
The profit-guarantee mechanism is the Bayes optimal mechanism against the
Pareto distribution with shape parameter: 
\begin{equation*}
\alpha =\frac{\eta }{\eta -1},
\end{equation*}%
and attains the infimum: 
\begin{equation}
\inf_{F}\frac{\Pi _{F}}{S_{F}}=\frac{\Pi _{P_{\alpha }}}{S_{P_{\alpha }}}%
\bigg|_{\alpha =\frac{\eta }{\eta -1}}=\frac{1}{\eta ^{\frac{\eta }{\eta -1}}%
}.  \label{eq:l}
\end{equation}
\end{theorem}

\begin{proof}
We can then compute the profit and social surplus in closed-form to conclude
that, when $\alpha >\eta /\eta -1$. 
\begin{equation*}
\frac{\Pi _{P_{\alpha }}}{S_{P_{\alpha }}}=\left( \frac{\alpha -1}{\alpha }%
\right) ^{\frac{\eta }{\eta -1}}.
\end{equation*}%
This equation is valid only when $\alpha <\eta /\eta -1$ because these are
the Pareto distributions for which social surplus remains finite. However,
by taking limit $\alpha \rightarrow \eta /(\eta -1)$, we get \eqref{eq:l}
(and we also consider a truncated distribution as in \eqref{trunc} and take
the limit $k\rightarrow \infty $).

Theorem \ref{prop:2} then follows from Theorem \ref{th:1}. If the seller can
guarantee herself a fraction $1/\eta ^{\frac{\eta }{\eta -1}}$ of the social
surplus, and this is in fact the best she can do for some distribution of
values, then this distribution of values minimizes the fraction of the
social surplus that the seller can generate as profit.
\end{proof}

Thus the Pareto distribution with shape $\alpha =\eta /(\eta -1)$ allows the
seller to generate the least amount of normalized profit. Finally, we can
analyze how the minimum normalized-profit varies with the cost elasticity.
It is easy to check that the normalized profit guarantee $1/\eta ^{\frac{%
\eta }{\eta -1}}$ is decreasing in $\eta .$ We can evaluate the minimum
profit for different cost elasticities:%
\begin{equation}
\lim_{\eta \rightarrow 1}\frac{1}{\eta ^{\frac{\eta }{\eta -1}}}=\frac{1}{e}%
;\ \ \ \lim_{\eta \rightarrow 2}\frac{1}{\eta ^{\frac{\eta }{\eta -1}}}=%
\frac{1}{4};\ \ \ \lim_{\eta \rightarrow \infty }\frac{1}{\eta ^{\frac{\eta 
}{\eta -1}}}=0.  \label{pr}
\end{equation}%
Note that the limit $\eta \rightarrow \infty $ corresponds to the case in
which the seller is selling an indivisible good. The limit $\eta \rightarrow
1$ corresponds to the case in which the seller has nearly\ constant marginal
cost.

We can similarly evaluate the consumer surplus at these different cost
elasticities and Pareto distribution and find that:%
\begin{equation}
\lim_{\eta \rightarrow 1}\frac{1}{\eta ^{\frac{1}{\eta -1}}}=\frac{1}{e};\ \
\ \ \lim_{\eta \rightarrow 2}\frac{1}{\eta ^{\frac{1 }{\eta -1}}}=\frac{1}{2}%
;\ \ \ \ \lim_{\eta \rightarrow \infty }\frac{1}{\eta ^{\frac{1}{\eta -1}}}%
=1.  \label{cs}
\end{equation}%
The minmax solution of profit guarantee mechanism and Pareto distribution
generate particular pairs of surplus sharing among seller and consumers.

\begin{corollary}[Consumer Surplus in the Profit Guarantee Mechanism]
\label{csc}\quad \newline
The profit-guarantee mechanism $M^{\ast }$ generates a constant consumer
surplus 
\begin{equation}
\frac{U_{M^{\ast }}}{S_{F}}=\frac{1}{\eta ^{\frac{1}{\eta -1}}}  \label{csc1}
\end{equation}%
across all\ distributions $F$.
\end{corollary}

We might have expected the uniformity of the profit guarantee across all
distributions from the minmax property of the mechanism. Indeed, in \cite%
{haro14}, the optimal single unit monopoly pricing policy also has the
property that it generates a uniform profit guarantee across all
distributions. By contrast, in the single unit monopoly pricing, the
consumer surplus share is not uniform across all distribution. For a given
willingness-to-pay $v$, the net utility of a buyer is $\left( v\ln v\right)
/\left( 1+\ln v\right) $ and thus the expected consumer surplus share thus
depends on the distribution of $v$. More precisely, in the profit-guarantee
mechanism, each consumer receives the same share of the efficient surplus.
By contrast, in the single unit monopoly pricing model, the share of the
consumer surplus is increasing in the willingness to pay.

The property of a uniform consumer surplus share is interesting in its own
right. But we might ask how the consumer surplus guarantee compares to
levels of consumer surplus that can be attained across all Bayes optimal
mechanisms. More generally, we can ask what the upper frontier of surplus
sharing is among seller and buyers in the nonlinear pricing environment.
This is what we pursue in the next Section \ref{sec:upper}.

\subsection{A Constant Mark-Up Mechanism}

It is useful to give an alternative representation of the optimal
profit-guarantee mechanism as an indirect mechanism. An incentive compatible
mechanism $\{Q(v),T(v)\}$ can be implemented by offering an indirect
mechanism that asks for a price $P\left( q\right) $ for a quality $q$.
Provided the indirect mechanism is sufficiently differentiable, the indirect
mechanism can be represented by its marginal price for quality, the
price-per-quality increment: 
\begin{equation*}
P^{\prime }\left( q\right) \triangleq p(q),
\end{equation*}%
given by: 
\begin{equation*}
p(q)=Q^{-1}(q).
\end{equation*}%
So that, if the buyer buys quality $q$, the total payment is: 
\begin{equation*}
P\left( q\right) =\int_{0}^{q}p(s)ds.
\end{equation*}%
In other words, the transfer is the integral of the price of each additional
quality increment.

\begin{corollary}[Optimal Profit-Guarantee Mechanism]
\label{cor:opt}\quad \newline
The profit-guarantee mechanism $M^{\ast }$ can be implemented by offering
quality increments $q\in \mathbb{R}$ at a price $p\left( q\right) $
satisfying: 
\begin{equation}
\frac{p\left( q\right) -c^{\prime }(q)}{p\left( q\right) }=\frac{\eta -1}{%
\eta }.  \label{li}
\end{equation}
\end{corollary}

$\allowbreak $Following this corollary, the optimal profit-guarantee
mechanism maintains a constant mark-up for each \emph{additional unit},
where \eqref{li} is called the \emph{Lerner's index.} The constant mark-up
here is determined by elasticity $\eta \,$of the cost function. %

It is informative to contrast the profit guarantee policy with the Bayesian
optimal policy for a given prior distribution $F$. In the Bayesian- optimal
mechanism the qualities are solved by the first-order condition with respect
to the virtual utility. Supposing for the moment that $F$ is regular, \cite%
{muro78} solve: 
\begin{equation*}
q(v)\in \arg \max_{q}\ \left\{ \left( v-\frac{1-F(v)}{f(v)}\right)
q-c(q)\right\} .
\end{equation*}%
The first-order condition is given by: 
\begin{equation*}
v-\frac{1-F(v)}{f(v)}-c^{\prime }(q(v))=0,
\end{equation*}%
and the incentive compatible transfers in the associated direct mechanism
are given by: 
\begin{equation*}
T(v)=vq(v)-\int_{0}^{v}q(s)ds.
\end{equation*}%
We thus have that: 
\begin{equation*}
T^{\prime }(v)=q^{\prime }(v)v.
\end{equation*}%
Thus, the marginal price per unit of quality is given by: 
\begin{equation*}
p(q(v))=\frac{T^{\prime }(v)}{q^{\prime }\left( v\right) }=v\text{.}
\end{equation*}%
Thus, the resulting markup is given by: 
\begin{equation}
\frac{p(q)-c^{\prime }(q)}{p(q)}=\frac{v-(v-\frac{1-F(v)}{f(v)})}{v}=\frac{%
1-F(v)}{f(v)v}.  \label{mu}
\end{equation}%
The right-hand-side is the negative of the reciprocal of the demand
elasticity: this is the classic formula for the Lerner's index. More
precisely, the demand for quality $q$ at any given per-unit-of-quality price 
$p(q(v))$ is: 
\begin{equation*}
D(p(q(v)))=1-F(v).
\end{equation*}%
Thus 
\begin{equation}
\frac{1-F(v)}{f(v)v}=-\frac{1}{\frac{dD}{dv}\frac{v}{1-F(v)}}.  \label{de}
\end{equation}%
We thus find the Bayes-optimal mechanism is determined entirely by the
demand elasticity which can be expressed in terms of the product of the
value $v $ and the hazard rate \ $f\left( v\right) /\left( 1-F\left(
v\right) \right) $. By contrast, the profit-guarantee mechanism is
determined only by elasticity of the cost function and does not refer to
neither the willingness-to-pay $v$ nor the distribution of the
willingness-to-pay. As the profit-guarantee is accomplished across all
possible distribution of values, it does not refer to any specific
distribution, but rather uses the cost information to offer a uniform menu
for all possible distribution of values.

\subsection{Non-Constant Cost Elasticity}

In Theorem \ref{th:1} and \ref{prop:2} we derive the optimal
profit-guarantee mechanism for constant elasticity cost functions. We can
then ask whether we can provide similar bounds in environment where the cost
function does not satisfy the constant elasticity condition. To this end, we
now consider convex costs function satisfying the curvature condition$\
c^{\prime \prime \prime }\left( q\right) \geq 0$. We denote the pointwise
cost elasticity by: 
\begin{equation*}
\eta (q)=\frac{dc(q)}{dq}\frac{q}{c(q)}.
\end{equation*}%
We assume that the cost elasticity is bounded: 
\begin{equation*}
\eta (q)\leq \bar{\eta}.
\end{equation*}%
We consider a mechanism in which the qualities supplied $\left( q\left(
v\right) \right) $\ are in a linear relationship to the marginal cost of
providing the quality: 
\begin{equation}
c^{\prime }(q(v))=zv,  \label{cm1}
\end{equation}%
where the linear parameter $z$ is chosen to satisfy: 
\begin{equation}
z=\frac{1}{\sqrt{\bar{\eta}-1}+1}.  \label{cm2}
\end{equation}%
Hence, we maintain a constant markup as in the profit-guarantee mechanism $%
M^{\ast }$. When the cost has constant cost elasticity, we have that the
quality given by $M^{\ast }$, see \eqref{q}, has the property that $z=1/\eta 
$.

\begin{proposition}[Profit Guarantee of Constant Markup Mechanism]
\label{prop:low}\quad \newline
The constant-markup selling strategy given (\ref{cm1})-(\ref{cm2}) generates
profits of at least: 
\begin{equation}
\Pi ^{z}=\frac{1}{\overline{\eta }+2\sqrt{\overline{\eta }-1}}S.  \label{cmm}
\end{equation}
\end{proposition}

\begin{proof}
We note that when the seller charges a constant markup, the profits
generated when the buyer's value is $v$ are: 
\begin{equation*}
\pi =\int_{0}^{q(v)}\frac{1}{z}c^{\prime }(s)-c^{\prime }(s)ds=\frac{1-z}{z}%
c(q(v)).
\end{equation*}%
Hence, the profits generated are proportional to the production cost of the
quality supplied to this type. We can write the profits in terms of the cost
elasticity as follows: 
\begin{equation*}
\pi =\frac{1-z}{z}\frac{zvq(s)}{\eta (q(v))}.
\end{equation*}%
Here we simply replaced the definition of the cost elasticity and we used
that $c^{\prime }(q(v))=zv$.

On the other hand, the efficient total surplus is given by: 
\begin{equation*}
S=\int_{0}^{q^{\ast }(v)}v-c^{\prime }(s)ds.
\end{equation*}%
We can write this as follows: 
\begin{align*}
S=& \int_{0}^{q(v)}v-c^{\prime }(s)ds+\int_{q(v)}^{q^{\ast }(v)}v-c^{\prime
}(s)ds \\
=& q(v)v-c(q(v))+\int_{q(v)}^{q^{\ast }(v)}v-c^{\prime }(s)ds.
\end{align*}%
We thus get that: 
\begin{equation*}
\frac{S}{\pi }=\frac{\eta (v)}{(1-z)}-\frac{z}{1-z}+\frac{z}{1-z}\frac{%
\int_{q(v)}^{q^{\ast }(v)}v-c^{\prime }(s)ds}{c(q(v)}.
\end{equation*}%
We now prove that: 
\begin{equation}
\frac{\int_{q(v)}^{q^{\ast }(v)}v-c^{\prime }(s)ds}{c(q(v))}\leq \frac{%
(1-z)^{2}}{z^{2}}.  \label{po1}
\end{equation}%
For this, we note that $c^{\prime \prime \prime }\geq 0$ implies that for
all $q$: 
\begin{equation*}
c^\prime(q)\geq c^{\prime }(q(v))+c^{\prime\prime }(q(v))(q-q(v)).
\end{equation*}
We thus get that: 
\begin{equation*}
c(q(v))\geq \frac{1}{2} \frac{(zv)^{2}}{c^{\prime \prime }(q(v))}.
\end{equation*}%
Analogously, $c^{\prime \prime \prime }\geq 0$ also implies that 
\begin{equation*}
\int_{q(v)}^{q^{\ast }(v)}v-c^{\prime }(s)ds\leq \frac{1}{2}\frac{%
((1-z)v)^{2}}{c^{\prime \prime }(q(v))}.
\end{equation*}%
Taking the ratio of both inequalities, we get \eqref{po1}. We thus conclude
that: 
\begin{equation*}
\frac{S}{\pi }=\frac{\eta (v)}{(1-z)}-\frac{z}{1-z}+\frac{(1-z)}{z}=\frac{%
1+z(\eta \left( v\right) -2)}{z(1-z)}.
\end{equation*}%
We can now take the reciprocal of this expression: 
\begin{equation*}
\frac{\pi }{S}=\frac{z(1-z)}{1+z(\eta (v)-2)}.
\end{equation*}%
We clearly have that the expression is decreasing in $\eta $. So, 
\begin{equation*}
\frac{\pi }{S}\leq \frac{z(1-z)}{1+z(\bar{\eta}-2)}\leq \frac{1}{\overline{%
\eta }+2\sqrt{\overline{\eta }-1}}.
\end{equation*}%
The second equality follows from maximizing the expression with respect to $%
z $, which has as a maximand 
\begin{equation*}
z=\frac{1}{\sqrt{\overline{\eta }-1}+1},
\end{equation*}%
which completes the proof.
\end{proof}

The constant mark-up mechanism given by (\ref{cm1})-(\ref{cm2}) can thus
provide a positive revenue guarantee in large class of convex cost functions
without requiring a constant elasticity. We can verify that the guarantee is
achieved by providing a larger quality than mechanism $M^{\ast }$ would if
the cost elasticity were equal to the upper bound $\overline{\eta }$ in the
relevant range, namely $\overline{\eta }>2$ as 
\begin{equation*}
\frac{1}{\sqrt{\overline{\eta }-1}+1}>\frac{1}{\overline{\eta }}\text{.}
\end{equation*}%
As the seller provides a larger consumer surplus to the consumer, the
revenue guarantee is weaker than in the optimal mechanism $M^{\ast }$
evaluated at the upper bound $\overline{\eta }:$%
\begin{equation*}
\frac{1}{\overline{\eta }+2\sqrt{\overline{\eta }-1}}<\frac{1}{\overline{%
\eta }^{\frac{\overline{\eta }}{\overline{\eta }-1}}}\text{.}
\end{equation*}%
Thus, the robustness against a larger class of cost functions with varying
elasticity is achieved by conceding surplus to the consumer against a lower
profit share.

\section{Boundaries of Surplus Sharing\label{sec:upper}}

The profit guarantee mechanism $M^{\ast }$ can secure a profit guarantee for
the seller as established by Theorem \ref{th:1} and \ref{prop:2}.
Surprisingly, the guarantee is not only a lower bound for some distribution
of values, but the mechanism enables the seller to attain this guarantee
uniformly across all distributions. Corollary \ref{csc} then showed that the
same mechanism also provides a uniform share of consumer surplus. But as the
profit-guarantee mechanism is chosen to attain the highest possible profit
level, we might be concerned that the profit guarantee mechanism is
succeeding in obtaining a high profit share by depressing or even minimizing
consumer surplus among all incentive compatible mechanism.

To approach this problem and answer this question, we now characterize the
upper frontier of the feasible consumer surplus and profit share across all
distributions: 
\begin{equation}
\sup_{F}\ \left\{ \frac{U_{F}}{S_{F}}:\quad \frac{\Pi _{F}}{S_{F}}=\beta
\right\} .  \label{eq:frontier}
\end{equation}%
We refer to the upper frontier as the surplus frontier. In other words, we
seek to identify the maximum consumer surplus given that the profit is
greater than or equal to some fraction $\beta \in \lbrack 0,1]$ of the
social surplus. Of course, problem \eqref{eq:frontier} is well defined when $%
\beta $ can be attained. In particular, when $\beta $ is the minimum profit
that can be attained (i.e., $\beta $ equals to \eqref{eq:minprof}) we will
find the maximum consumer surplus across all distribution of values.

\subsection{Surplus Frontier}

We now provide a complete description of the surplus frontier.

\begin{proposition}[Surplus Frontier]
\label{prop:fornt}\quad \newline
The surplus frontier is given by: 
\begin{equation}
\sup_{F}\left\{ \frac{U_{F}}{S_{F}}:\quad \frac{\Pi _{F}}{S_{F}}=\beta
\right\} =\frac{\eta }{\eta -1}\left( \beta ^{\frac{1}{\eta }}-\beta \right)
.  \label{eq:sold}
\end{equation}%
The constraint is feasible if and only if $\beta \in \left[ 1/\eta ^{\frac{%
\eta }{\eta -1}},1\right] $.
\end{proposition}

\begin{proof}
We begin by writing $U_{F}$ and $\Pi _{F}$ explicitly. We denote the virtual
values as follows: 
\begin{equation*}
\phi (v)\triangleq v-\frac{1-F(v)}{f(v)}.
\end{equation*}%
We denote the ironed virtual values by $\bar{\phi}.$ We can find a
collection of intervals $\{(\underline{v}_{i},\bar{v}_{i})\}_{i\in I}$ such
that $\phi (v)=\bar{\phi}(v)$ in $[\underline{v}_{i},\bar{v}_{i}]$ and
outside these intervals (i.e., in each interval of the form $(\bar{v}_{i},%
\underline{v}_{i+1}))$ we have that $\phi(v)<\bar{\phi}(v)$ and $\bar{\phi}%
(v)$ remains constant. The ironed-virtual values outside these intervals are
given by: 
\begin{equation*}
\bar{\phi}(v)=\frac{\int_{\bar{v}_{i}}^{\underline{v}_{i+1}}(v-\frac{1-F(v)}{%
f(v)})f(v)dv}{F(\underline{v}_{i+1})-F(\bar{v}_{i})}.
\end{equation*}%
The optimal quality offered by the optimal mechanism is given by (see \cite%
{toik11}): 
\begin{equation*}
Q(v)=\max \{\frac{\eta -1}{\eta }\bar{\phi}(v)^{\frac{\eta }{\eta -1}},0\}.
\end{equation*}%
As it is standard in the literature, the quality is constant in $(\bar{v}%
_{i},\underline{v}_{i+1})$, so to avoid confusion we write: 
\begin{equation*}
q_{i}\triangleq Q(\bar{v}_{i})=Q(\underline{v}_{i+1}).
\end{equation*}%
And types with a negative virtual value will be excluded. We then have that: 
\begin{align}
U_{F}=& \int_{0}^{\infty }\bar{\phi}(v)^{\frac{1}{\eta -1}}(1-F(v))dv;
\label{expr2} \\
\Pi _{F}=& \frac{\eta -1}{\eta }\int_{0}^{\infty }\bar{\phi}(v)^{\frac{\eta 
}{\eta -1}}f(v)dv.  \notag
\end{align}%
Finally, we note that the first-best surplus is given by: 
\begin{equation*}
S_{F}=\int_{0}^{\infty }\frac{\eta -1}{\eta }v^{\frac{\eta }{\eta -1}}f(v)dv.
\end{equation*}%
This corresponds to solving \eqref{eq:ts} explicitly.

We now note that the normalized profit can be written as follows: 
\begin{equation}
\Pi _{F}=\frac{\eta -1}{\eta }\int_{0}^{\infty }\bar{\phi}(v)^{\frac{1}{\eta
-1}}(v-\frac{1-F(v)}{f(v)})f(v)dv.  \label{fe}
\end{equation}%
To verify this, we note that in any regular interval $[\underline{v}_{i},%
\bar{v}_{i}]$ we have that 
\begin{equation*}
\phi (v)=\bar{\phi}(v)=(v-\frac{1-F(v)}{f(v)}).
\end{equation*}%
Hence, we have that: 
\begin{equation*}
\bar{\phi}(v)^{\frac{1}{\eta -1}}(v-\frac{1-F(v)}{f(v)})=\bar{\phi}(v)^{%
\frac{\eta }{\eta -1}}.
\end{equation*}%
In any non-regular interval $[\bar{v}_{i},\underline{v}_{i+1}]$ we have that 
\begin{equation*}
\bar{\phi}(v)=\frac{\int_{\bar{v}_{i}}^{\underline{v}_{i+1}}(v-\frac{1-F(v)}{%
f(v)})f(v)dv}{F(\underline{v}_{i+1})-F(\bar{v}_{i})}.
\end{equation*}%
Hence, we have that: 
\begin{equation*}
\int_{\bar{v}_{i}}^{\underline{v}_{i+1}}\bar{\phi}(v)^{\frac{1}{\eta -1}}(v-%
\frac{1-F(v)}{f(v)})dv=\int_{\bar{v}_{i}}^{\underline{v}_{i+1}}\bar{\phi}%
(v)^{\frac{\eta }{\eta -1}}dv.
\end{equation*}%
We thus prove that \eqref{fe} is satisfied and we can write the normalized
profit as follows: 
\begin{equation}
U_{F}=\int_{0}^{\infty }\bar{\phi}(v)^{\frac{1}{\eta -1}}vf(v)dv-\frac{\eta 
}{\eta -1}\Pi _{F}  \label{us}
\end{equation}%
%
%
%
%
%
%
%
%
%
%
%
%
%
%
%
%
%
%
%
Using H\"{o}lder's inequality we get that: 
\begin{equation*}
\int_{0}^{\infty }\bar\phi(v)^{\frac{1}{\eta -1}}vf(v)dv\leq \left(
\int_{0}^{\infty }\bar{\phi}(v)^{\frac{\eta }{\eta-1}}f(v)dv\right) ^{\frac{1%
}{\eta }}\left( \int_{0}^{\infty }v^{\frac{\eta }{\eta -1}}f(v)dv\right) ^{%
\frac{\eta -1}{\eta }}.
\end{equation*}%
We thus have that: 
\begin{equation*}
U_{F}\leq \frac{\eta }{\eta -1}\left( \left( \Pi _{F}\right) ^{\frac{1}{\eta 
}}\left( S_{F}\right) ^{\frac{\eta -1}{\eta }}-\Pi _{F}\right) .
\end{equation*}%
Dividing by $S_{F}$ we obtain an upper bound on normalized consumer surplus: 
\begin{equation}
\frac{U_{F}}{S_{F}}\leq \frac{\eta }{\eta -1}\left( \left( \frac{\Pi _{F}}{%
S_{F}}\right) ^{\frac{1}{\eta }}-\frac{\Pi _{F}}{S_{F}}\right) .
\label{eq:d}
\end{equation}%
We note that this function $h(x)=x^{1/\eta }-x$ is decreasing in $x$ for all 
$x\in \lbrack 1/\eta ^{\frac{\eta }{\eta -1}},1]$. Proposition \ref{prop:2}
states that $\Pi _{F}/S_{F}\geq 1/\eta ^{\frac{\eta }{\eta -1}}$, so we
obtain that the right-hand-side of \eqref{eq:sold} is an upper bound. We now
need to show the inequality \eqref{eq:d} is tight.

To prove the inequality is tight, consider a Pareto distribution $%
F(v)=1-v^{-\alpha }$ with $\alpha =\frac{1}{1-\beta ^{\frac{\eta -1}{\eta }}}
$. We get that: 
\begin{equation}
\frac{\Pi _{F}}{S_{F}}=\beta \quad \text{and}\quad \frac{U_{F}}{S_{F}}=\frac{%
\eta }{\eta -1}(\beta ^{\frac{1}{\eta }}-\beta ).  \label{eq:rmin}
\end{equation}%
This proves that \eqref{eq:d} is tight. 
Note that replacing $\Pi _{F}/S_{F}=1/\eta ^{\frac{\eta }{\eta -1}}$ in %
\eqref{eq:sold} we obtain: 
\begin{equation*}
\frac{U_F}{S_F}=\frac{1}{\eta ^{\frac{1}{\eta -1}}}.
\end{equation*}
\end{proof}

We illustrate the frontier for different values of $\eta$ in Figure \ref%
{fig11}. As a direct corollary, we have the following upper bound on
consumer surplus.

\begin{corollary}[Maximum Consumer Suplus]
\label{cor:b}\quad \newline
The consumer surplus is bounded above as follows, 
\begin{equation}
\sup_{F}\frac{U_{F}}{S_{F}}=\frac{1}{\eta ^{\frac{1}{\eta -1}}}
\label{eq:l2}
\end{equation}%
and is attained by the Pareto distribution with shape parameter 
\begin{equation*}
\alpha =\frac{\eta }{\eta -1}.
\end{equation*}
\end{corollary}

\begin{figure}[H]
	\centering
	\includegraphics[width=0.67\textwidth]{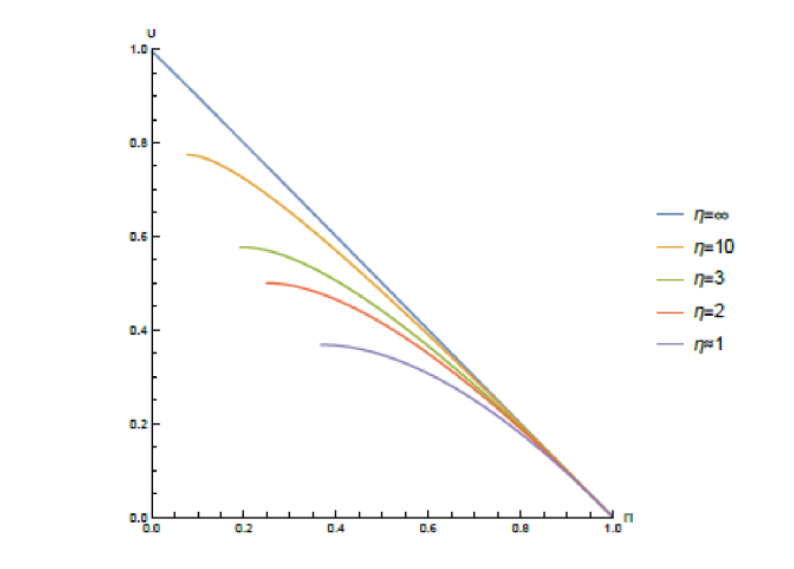}
	\caption{Upper Frontier for Different Iso-Elastic Cost Functions}
	\label{fig11}
\end{figure}

Hence, we obtain that in general the buyer can capture a greater share of
the social surplus when the cost is more elastic.

Surprisingly, the mechanism $M^{\ast }$ that guarantees the seller the
highest profit guarantee across all distributional environments is also the
mechanism that offers the buyer the highest expected consumer surplus across
all optimal mechanism for all distributional environments. Thus, the
mechanism that guarantees the seller the highest revenue does so by
conceding the most consumer surplus and offering a nearly efficient
mechanism. It provides an equilibrium allocation to every agent that is a
constant fraction of the socially efficient allocation.

\subsection{Lower Bound on Social Surplus}

We now find a lower bound on the total surplus generated by a
Bayesian-optimal mechanism across all distribution of values. That is, we
find the minimum social surplus: 
\begin{equation}
\inf_{F}\ \frac{U_{F}+\Pi _{F}}{S_{F}}.  \tag{D}  \label{eq:minsurplus}
\end{equation}%
Note that $S_{F}$ is the social surplus generated by the efficient
allocation so, in general, $U_{F}+\Pi _{F}<S_{F}$. We will be able to find
this lower bound only when $\eta \geq 2$, that is, when the marginal cost is
convex.

We now provide a lower bound on the distortions generated by any mechanism.
Before we provide the result, we note that: 
\begin{equation*}
\frac{U_{P_{\alpha }}}{S_{P_{\alpha }}}\bigg|_{\alpha =1}=0\quad \text{and}%
\quad \frac{\Pi _{P_{\alpha }}}{S_{P_{\alpha }}}\bigg|_{\alpha =1}=\frac{1}{%
\eta }.
\end{equation*}%
That is, when the distribution of values is the Pareto distribution with
shape parameter $\alpha =1$ the consumer's surplus is 0 (in fact, it is 0
for every truncated Pareto distribution $P_{1 ,k}$, not only in the limit),
and the normalized profit is $1/\eta .$ We thus have that: 
\begin{equation}
\frac{U_{P_{\alpha }}+\Pi _{P_{\alpha }}}{S_{P_{\alpha }}}\bigg|_{\alpha =1}=%
\frac{1}{\eta }.  \label{eq:rmin2}
\end{equation}%
In other words, the generated social surplus is a fraction $1/\eta $ of the
efficient social surplus.

\begin{proposition}[Lower Bound on Social Surplus]
\label{lwb}\quad \newline
When $\eta \geq 2$, social surplus is bounded below by: 
\begin{equation*}
\inf_{F}\ \frac{U_{F}+\Pi _{F}}{S_{F}}=\frac{U_{P_{\alpha }}+\Pi _{P_{\alpha
}}}{S_{P_{\alpha }}}\bigg|_{\alpha =1}=\frac{1}{\eta }.
\end{equation*}
\end{proposition}

\begin{proof}
Following expression \eqref{us} we obtain: 
\begin{equation*}
\frac{\eta -1}{2\eta -1}U_F+\frac{\eta }{2\eta -1}\Pi_F =\frac{\eta -1}{%
2\eta -1}\int_{0}^{\infty }\bar{\phi}(v)^{\frac{1}{\eta -1}}vf(v)dv.
\end{equation*}%
We first note that, $\bar{\phi}(v)\leq v$, so when $\eta \geq 2$, we have
that: 
\begin{equation*}
\bar{\phi}(v)^{\frac{1}{\eta -1}}v\geq \bar{\phi}(v)v^{\frac{1}{\eta -1}}.
\end{equation*}%
We now note that: 
\begin{equation}  \label{ded}
\int_{0}^{\infty }\bar{\phi}(v)v^{\frac{1}{\eta -1}}f(v)dv\geq
\int_{0}^{\infty }\phi (v)v^{\frac{1}{\eta -1}}f(v)dv=\frac{1}{\eta }\int v^{%
\frac{\eta -1}{\eta }}f(v)dv.
\end{equation}%
The inequality follows from the fact that, by construction of the ironed
virtual values, for any increasing function $h(v)$ we have that $%
\int_{0}^{\infty }\bar{\phi}(v)h(v)f(v)dv\geq \int_{0}^{\infty }\phi
(v)h(v)f(v)dv$ (see \cite{klms21} ). The equality follows from integrating
by parts. Then the right-hand-side of \eqref{ded} is exactly equal to $S_{F}$%
, so we have that: 
\begin{equation}
\frac{\eta -1}{2\eta -1}U_{F}+\frac{\eta }{2\eta -1}\Pi _{F}\geq \frac{1}{%
2\eta -1}S_{F}.  \label{eq:lb}
\end{equation}%
We can now show the inequality is tight. Using \eqref{eq:rmin2} we get that: 
\begin{equation*}
\frac{\eta -1}{2\eta -1}U_{F}+\frac{\eta }{2\eta -1}\Pi _{F}=\frac{1}{2\eta
-1}S_{F}.
\end{equation*}%
This corresponds to the lower bound \eqref{eq:lb}. We thus have that: 
\begin{equation*}
\inf_{F}\ \frac{\frac{\eta -1}{2\eta -1}U_{F}+\frac{\eta }{2\eta -1}\Pi _{F}%
}{S_{F}}=\frac{\eta }{2\eta -1}\frac{1}{\eta }.
\end{equation*}%
Since $\frac{\eta }{2\eta -1}>1/2$ and the infimum is attained at an
information structure that generates 0 consumer surplus, we must also have
that: 
\begin{equation*}
\inf_{F}\ \frac{\frac{1}{2}U_{F}+\frac{1}{2}\Pi _{F}}{S_{F}}=\frac{1}{2}%
\frac{1}{\eta }.
\end{equation*}%
Multiplying by 2, we obtain the result.
\end{proof}

As the cost becomes more inelastic ($\eta $ grows), the lower bound becomes
smaller. When the cost is quadratic, then the optimal mechanism always
generates at least 1/2 of the social surplus. Note that the result obviously
does not apply for every $\eta <2$. In particular, in the limit $\eta
\rightarrow 1$ we know that the optimal mechanism might introduce
non-negligible distortions (see Proposition \ref{prop:fornt}).%

\subsection{Complete Surplus Boundary}

Finally, we might be interested in a complete characterization of the set of
feasible surplus pairs. For the case of the quadratic cost function we can
provide such a description. The reason that $\eta =2$ is particularly easy
to analyze is that in this case we can compute all qualities in closed-form.

Beyond the quadratic case, we can also provide a general lower bound for the
equilibrium surplus realized relative to the social surplus. This general
bound only requires that the marginal cost is convex, thus $\eta \geq 2$.

The feasible normalized profit and consumer surplus are: 
\begin{equation*}
\mathcal{F}=\{(x,y)\in \mathbb{R}\mid \text{there exists $F$ such that }x=%
\frac{U_{F}}{S_{F}}\text{ and }y=\frac{\Pi _{F}}{S_{F}}\}.
\end{equation*}%
%
%
%
%
%
%
%
%
%
%
%
%
%
%
%
%
%
%
We consider the problem of characterizing $\mathcal{F}$. We will only be
able to do this when $\eta =2$, but previous results suggest that an
analogous characterization applies to all $\eta .$

When $\eta =2,$ for a fixed $k$, we have that: 
\begin{equation*}
\frac{U_{^{P_{\alpha ,k}}}}{S_{P_{\alpha ,k}}}=\frac{2(\alpha -1)(k^{\alpha
-2}-1)}{\alpha (\alpha k^{\alpha -2}-2)}\quad \text{and}\quad \frac{\Pi
_{P_{\alpha ,k}}}{S_{P_{\alpha ,k}}}=\frac{(\alpha -1)^{2}k^{\alpha -2}-1}{%
\alpha (\alpha k^{\alpha -2}-2)}.
\end{equation*}%
Taking the limit $k\rightarrow \infty $, we get that: 
\begin{equation}
(\frac{U_{P_{\alpha }}}{S_{P_{\alpha }}},\frac{\Pi _{P_{\alpha }}}{%
S_{P_{\alpha }}})=%
\begin{cases}
\left( \frac{2(\alpha -1)}{\alpha ^{2}},\frac{(\alpha -1)^{2}}{\alpha ^{2}}%
\right) , & \text{for }\alpha \in \lbrack 2,\infty ); \\ 
(\frac{\alpha -1}{\alpha },\frac{1}{2\alpha }), & \text{for }\alpha \in
\lbrack 1,2].%
\end{cases}
\label{eq:bound}
\end{equation}%
The curve for $\alpha \geq 2$ is the one characterized in Proposition \ref%
{prop:fornt}; the curve for $\alpha \in \lbrack 1,2]$ does not have a direct
counter part in the results we have provided thus far (except for $\alpha =1$%
).

\begin{proposition}[Feasible Normalized Profits and Utilities]
\label{th:2}\quad \newline
The closure of $\mathcal{F}$ is given by the area enclosed by the curves in %
\eqref{eq:bound}. Every point in the interior of $\mathcal{F}$ is generated
by some truncated Pareto distribution $P_{\alpha ,k}$.
\end{proposition}

\begin{proof}
The upper boundary was characterized in Proposition \ref{prop:fornt}. We now
characterize the lower boundary of $\mathcal{F}$. Writing \eqref{eq:lb} for $%
\eta=2$, we get: 
\begin{equation*}
\frac{1}{3}U_{F}+\frac{2}{3}\Pi _{F}\geq \frac{1}{3}S_{F}.
\end{equation*}%
However, the limit of Pareto distributions with parameter $\alpha \in
\lbrack 1,2]$ (see \eqref{eq:rmin}) gives exactly that: 
\begin{equation*}
\frac{\frac{1}{3}U_{P_{\alpha }}+\frac{2}{3}\Pi _{P_{\alpha }}}{S_{P_{\alpha
}}}=\frac{1}{3}.
\end{equation*}%
Hence, these distributions give the lower frontier of the set of feasible
consumer surplus and profit.
\end{proof}

The proposition gives a full characterization of the set of normalized
profit and consumer surplus generated by any distribution of values. We
provide an illustration of the result below in Figure \ref{figp}. We
established that the upper boundary of the surplus sharing between seller
and buyers is attained by the family of Pareto distributions with $\alpha
\geq 2$. The lower bound along the segment that provides positive consumer
surplus is attained by Pareto distribution with $1\leq \alpha \leq 2$.
Finally, the lower segment with zero surplus is attained by truncated Pareto
distribution. Here, the seller offers the product only to the buyers with a
value in the mass point of the truncated distribution. As a consequence, the
seller can extract all the surplus and provides the efficient allocation for
those buyers at the upper mass point. The buyers who get served receive zero
net surplus.

Thus, the family of truncated Pareto distribution generates all feasible
surplus pairs. Yet we may be interested in how other distribution of value
may affect the generation and distribution of surplus between the seller and
the buyers. In Figure \ref{figbp} we illustrate the range of outcomes that
is generated by other families of distribution including binary, uniform and
power distributions.

\begin{figure}[H]
	\centering
	\includegraphics[width=0.85\textwidth]{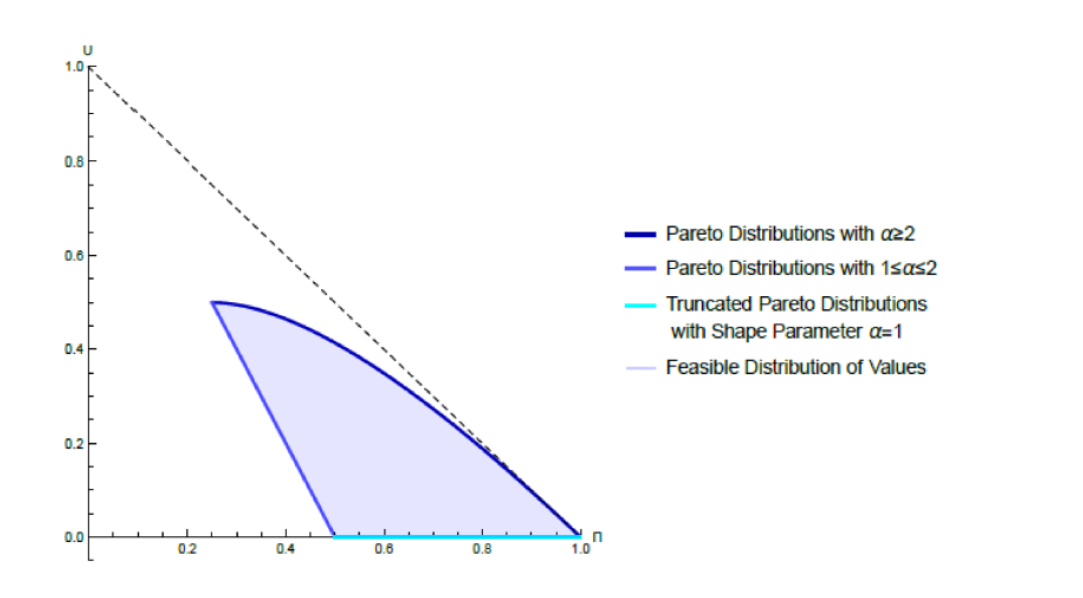}
	\caption{Feasible Normalized Consumer Surplus and Normalized Profits for Quadratic Cost}
	\label{figp}
\end{figure}


\begin{figure}[H]
	\centering
	\includegraphics[width=0.85\textwidth]{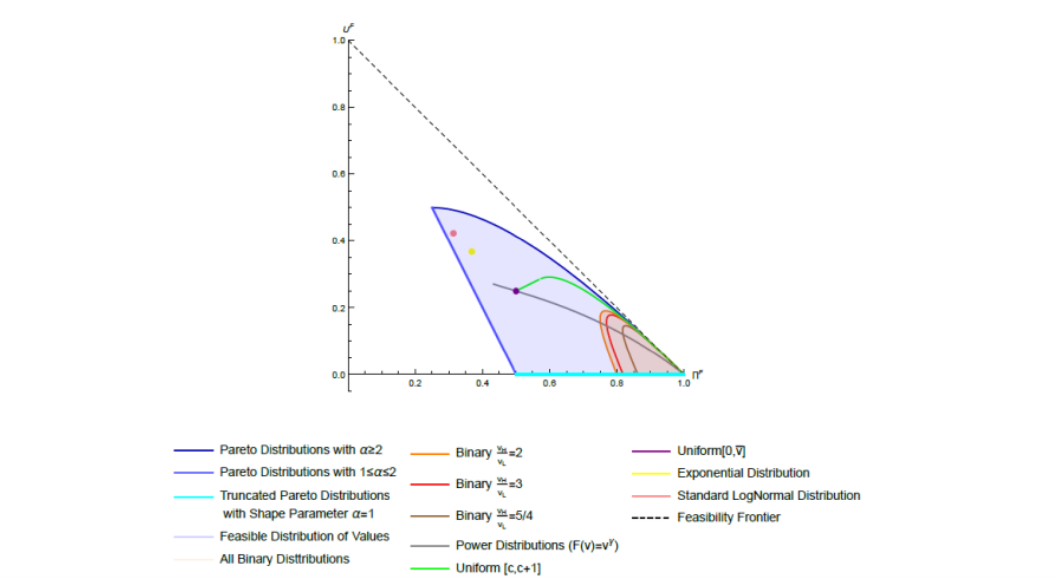}
	\caption{Illustration of the Normalized Consumer Surplus and Normalized Profits for Different Distributions with Quadratic Costs}
	\label{figbp}
\end{figure}


\section{Quantity Discrimination\label{sec:qua}}

So far, we considered a model of quality discrimination in the spirit of 
\cite{muro78}. We now investigate corresponding results for quantity
discrimination in the spirit of \cite{mari84}. We first consider the
multiplicatively separable model and then develop more general results. In
the model of quality discrimination of \cite{muro78}, the buyer has a
constant marginal willingness to pay for quality, and the cost of quality
provision is increasing and convex. By contrast, in the model of quantity
discrimination, there is a constant marginal cost of providing an additional
unit of a given product. The diminishing returns now arise from the
concavity of the utility function. Correspondingly, in the former the payoff
environment is described by the cost elasticity whereas in the later the
demand elasticity determines the choices of buyer and seller.

We first provide a profit guarantee for the case of multiplicatively
separable utility functions and then extend it to nonlinear environments
without separability conditions.

\subsection{Multiplicatively Separable Utility}

We now assume that the utility function is given by 
\begin{equation*}
u(v,q)=v\frac{\eta }{\eta +1}q^{\frac{\eta +1}{\eta }},
\end{equation*}%
for some $\eta \in (-\infty ,-1)$. Thus the utility function for a higher
quantity is increasing and concave. The cost of production is linear $%
c(q)=cq $, where we normalize $c=1.$

The demand is defined by the inverse of the marginal utility:%
\begin{equation*}
D(v,p)\triangleq u_{q}^{-1}(v,p),
\end{equation*}%
where the subscript $q$\ denotes the partial derivative with respect to $q$.
With the above parametrization we find that the demand elasticity is 
\begin{equation*}
\frac{\partial D(v,p)}{\partial p}\frac{p}{D(v,p)}\triangleq \eta .
\end{equation*}%
We note that as we shift from cost elasticity to demand elasticity, we
maintain $\eta $ as the parameter of the elasticity. However, now $\eta $ is
a negative number $\eta \in (-\infty ,-1)$.

As earlier in the case of quality discrimination, the Pareto distribution
with the shape parameter $\alpha \in (1,\infty )$ is playing a critical role
for the minmax problem.

\begin{theorem}[Profit Guarantee with Quantity Discrimination]
\label{thm:opti}\quad \newline
The profit guarantee mechanism is a uniform-price mechanism $t=p^{\ast }q$
with 
\begin{equation}
p^{\ast }=\frac{{\eta }}{{\eta +1}}.  \label{up}
\end{equation}%
It generates profits: 
\begin{equation}
\Pi ^{\ast }= \left( \frac{{\eta }}{{\eta }+1}\right) ^{{\eta }}S,
\label{up1}
\end{equation}%
for every $F$. Furthermore, The profit-guarantee mechanism is the Bayes
optimal mechanism against the Pareto distribution with shape parameter $%
\alpha=|\eta|$ and attains the infimum: 
\begin{equation}
\lim_{\alpha \rightarrow \left\vert \eta \right\vert }\frac{\Pi _{P_{\alpha
}}}{S_{P_{\alpha }}}=\left( \frac{\eta }{\eta +1}\right) ^{\eta }.
\label{eq:dd2}
\end{equation}%
%
%
%
%
%
%
%
%
%
%
%
%
%
%
%
%
\end{theorem}

%
%

\begin{proof}
In the baseline model of Section \ref{sec:model}, we have a buyer with
utility function: 
\begin{equation*}
u(v,q,t)=vq-t,
\end{equation*}%
and a seller with a cost function: 
\begin{equation*}
c(q)=\frac{q^{\eta }}{\eta }.
\end{equation*}%
Now consider the following change of variables 
\begin{equation*}
\hat{q}\triangleq \frac{q^{\eta }}{\eta },\ \ \hat{\eta}\triangleq -\frac{%
\eta }{\eta -1},\ \ \hat{v}\triangleq v(\frac{\hat{\eta}}{\hat{\eta}+1})^{%
\frac{1}{\hat{\eta}}}.
\end{equation*}%
We then have that the utility functions and cost functions are given by: 
\begin{equation*}
u(v,q,t)=\hat{v}(\frac{\hat{\eta}}{\hat{\eta}+1})\hat{q}^{\frac{\hat{\eta}+1%
}{\hat{\eta}}}-t,
\end{equation*}%
and 
\begin{equation*}
c(q)=\hat{q}.
\end{equation*}%
With this change of variable we then obtain the model of this section. The
profit guarantee result of Theorem \ref{th:1} and \ref{prop:2} then follow
immediately. We then establish that the uniform price mechanism is Bayes
optimal against the Pareto distribution with parameter $\alpha $.

The social surplus is given by: 
\begin{align*}
S_{P_{\alpha }}& =\int_{1}^{\infty }\max_{q}\{v\frac{\eta }{\eta +1}q^{\frac{%
\eta +1}{\eta }}-q\}dF(v) \\
& =\int_{1}^{\infty }\frac{-1}{\eta +1}v^{-\eta }dF(v) \\
& =\int_{1}^{\infty }\frac{-1}{\eta +1}v^{-\eta }dF(v) \\
& =\frac{-1}{\eta +1}\frac{-1}{-\eta +\alpha }.
\end{align*}%
On the other hand, the profit are given by: 
\begin{align*}
\Pi _{P_{\alpha }}& =\int_{1}^{\infty }\max_{q}\{v\frac{\eta }{\eta +1}q^{%
\frac{\eta +1}{\eta }}-q-\frac{\eta }{\eta +1}q^{\frac{\eta +1}{\eta }}\frac{%
1-F(v)}{f(v)}\}dF(v) \\
& =\int_{1}^{\infty }\max_{q}\{v\frac{\eta }{\eta +1}\frac{\alpha -1}{\alpha 
}q^{\frac{\eta +1}{\eta }}-q\}dF(v) \\
& =\int_{1}^{\infty }\frac{-1}{\eta +1}\left( \frac{\alpha +1}{\alpha }%
\right) ^{-\eta }v^{-\eta }dF(v) \\
& =\frac{1}{\eta +1}\frac{1}{\alpha -\eta }\left( \frac{\alpha +1}{\alpha }%
\right) ^{-\eta }
\end{align*}%
We then have that: 
\begin{equation*}
\frac{\Pi_{P_{\alpha }} }{S_{P_{\alpha }}}=\left( \frac{\alpha +1}{\alpha }%
\right) ^{-\eta }.
\end{equation*}%
Taking the limit $\alpha \rightarrow \left\vert \eta \right\vert $, we
obtain the result.
\end{proof}

Thus, in the case of concave utility functions and linear cost functions we
can recover a profit guarantee mechanism. The mechanism maintains the
constant mark-up property that we saw earlier, but in the presence of linear
costs, we now have that a linear pricing mechanism, a uniform per unit
price, generates the profit guarantee. Thus, the profit guarantee can be
established with an even simpler mechanism. With the change in variable
suggested in the proof of Theorem \ref{thm:opti}, it also follows
immediately that the profit guarantee mechanism generates a uniform profit
and consumer surplus share for all distributions $F$. Thus, the profit
guarantee mechanism maintains a uniform sharing of surplus between buyers
and seller across all distributions.

\subsection{Nonlinear Utility}

We now consider a class of nonlinear utility functions in which
willingness-to-pay and quantity can interact in a nonlinear manner and
without the former multiplicative separability condition. Thus, we assume
that the utility net of the payment $t\in \mathbb{R}_{+}$ is: 
\begin{equation*}
u(v,q,t)=h(v,q)-t,
\end{equation*}%
where $h$ is concave in $q$ given $v$. The willingness-to-pay parameter $v$
remains distributed according to $F$ and the cost of production remains
linear $c(q)=cq$ and we normalize $c=1$ without loss of generality.

With the nonlinearity now appearing in the utility function, we have
implicitly assumed that the seller has perfect information about the exact
shape of the nonlinearity while assuming incomplete information about the
willingness-to-pay $v$\ of the buyer. Therefore we are now asking whether we
can find a profit guarantee for the seller even when the seller has
imperfect information about the nonlinearity and hence the elasticity of the
demand function.

The demand function is then defined by the inverse of the marginal utility: 
\begin{equation*}
D(v,p)\triangleq h_{q}^{-1}(v,p),
\end{equation*}%
where the subscript $q$\ denotes the partial derivative with respect to $q$.
The demand elasticity is then given by: 
\begin{equation*}
\eta (v,p)\triangleq \frac{\partial D(v,p)}{\partial p}\frac{p}{D(v,p)},
\end{equation*}%
where the demand elasticity is assumed to be negative $\eta (v,p)<0$ for all 
$v,p$. We assume that, for all $p\in \lbrack 1,\infty ]$, 
\begin{equation}
\text{$\eta (v,p)$ is non-increasing in }p\text{\ and $\eta (v,p)\in \lbrack 
\bar{\eta}-1,\bar{\eta}]$,}  \label{eta}
\end{equation}%
for some $\bar{\eta}\in (-\infty ,-1)$. \ We next present a robust profit
guarantee that holds as long as the demand elasticity $\eta (v,p)$\ is
within the range $[\bar{\eta}-1,\bar{\eta}]$ for some upper bound $\bar{\eta}%
<-1$.

For a given demand function $D\left( v,p\right) $ , the optimal uniform
price is given by: 
\begin{equation*}
\hat{p}=\arg \max_{p}D(v,p)(p-c).
\end{equation*}%
The first-order condition can be written as follows: 
\begin{equation*}
\hat{p}=c\frac{\eta (v,\hat{p})}{\eta (v,\hat{p})+1}.
\end{equation*}%
We then have that: 
\begin{equation*}
\hat{p}\leq c\frac{\bar{\eta}}{\bar{\eta}+1}.
\end{equation*}%
Since the upper bound will be relevant for our analysis, it is useful to
denote by $\Pi ^{\ast }$ the profit generated by the uniform price mechanism
with price $p^{\ast }$: 
\begin{equation*}
\Pi ^{\ast }\triangleq D(v,p^{\ast })(p^{\ast }-c).
\end{equation*}

\begin{theorem}[Robust Profit-Guarantee Mechanism]
\label{thm:opt}\quad \newline
The uniform-price mechanism $t=p^{\ast }q$, where 
\begin{equation*}
p^{\ast }=\frac{{\bar{\eta}}}{{\bar{\eta}+1}}
\end{equation*}
guarantees a profit share of the social surplus: 
\begin{equation}
\Pi ^{\ast }\geq \left( \frac{\bar{\eta}}{\bar{\eta}+1}\right) ^{\bar{\eta}%
}S.  \label{optq}
\end{equation}
\end{theorem}

\begin{proof}
The profit generated by a uniform price mechanism is given by: 
\begin{equation*}
\Pi ^{\ast }=D(v,p^{\ast })(p^{\ast }-c).
\end{equation*}%
The social surplus is given by: 
\begin{equation*}
S=\int_{c}^{\infty }D(v,p)dp.
\end{equation*}%
The demand satisfies: 
\begin{equation*}
\log (\frac{D(v,p)}{D(v,p^{\ast })})=\int_{p^{\ast }}^{p}\frac{\eta (v,s)}{s}%
ds.
\end{equation*}%
Since the price elasticity is non-increasing, we have that 
\begin{equation*}
\log (\frac{D(v,p)}{D(v,p^{\ast })})=\int_{p^{\ast }}^{p}\frac{\eta (v,s)}{s}%
ds\leq \eta (v,p^{\ast })\int_{p^{\ast }}^{p}\frac{1}{s}ds={\eta (v,p^{\ast
})}\log \left( \frac{p}{p^{\ast }}\right) .
\end{equation*}%
We thus have that: 
\begin{equation*}
D(v,p)\leq D(v,p^{\ast })\left( \frac{p}{p^{\ast }}\right) ^{\eta (v,p^{\ast
})}.
\end{equation*}%
We then have that: 
\begin{equation*}
S\leq D(v,p^{\ast })\left( \frac{1}{p^{\ast }}\right) ^{\eta (v,p^{\ast
})}\int_{c}^{\infty }p^{\eta (v,p^{\ast })}dp=D(v,p^{\ast })\left( \frac{1}{%
p^{\ast }}\right) ^{\eta (v,p^{\ast })}\frac{-1}{\eta (v,p^{\ast })+1}%
c^{\eta (v,p^{\ast })+1}.
\end{equation*}%
We then have that 
\begin{equation*}
\frac{\Pi ^{\ast }}{S}\geq \frac{-(\eta (v,p^{\ast })+1)(p^{\ast })^{\eta
(v,p^{\ast })}(p^{\ast }-c)}{c^{\eta (v,p^{\ast })+1}}.
\end{equation*}%
We now note that the function $g(\eta )\triangleq -(\eta +1)p^{\eta }$ is
quasi-concave in $\eta $. We also have that we assumed that $\eta (v,p^{\ast
})\in \lbrack \underline{\eta },\bar{\eta}]$ and: 
\begin{equation*}
-(\bar{\eta}+1)(p^{\ast })^{\bar{\eta}}=-(\underline{\eta }+1)(p^{\ast })^{%
\underline{\eta }}.
\end{equation*}%
where $\underline{\eta }=\bar{\eta}-1$. We thus have that: 
\begin{equation*}
\frac{\Pi ^{\ast }}{S}\geq \frac{-(\bar{\eta}+1)(p^{\ast })^{\bar{\eta}%
}(p^{\ast }-c)}{c^{\eta (v,p^{\ast })+1}}=\left( \frac{\bar{\eta}}{\bar{\eta}%
+1}\right) ^{\bar{\eta}}.
\end{equation*}%
Now, since the bound was established pointwise for every $v$, it holds in
aggregate across all $v$.
\end{proof}

Theorem \ref{thm:opt} gives a profit guarantee for an environment where the
demand elasticity may vary within a limited range across willingness-to-pay
and price. Thus in contrast to the earlier results, it does not require a
constant demand elasticity. The robustness of the profit guarantee is
perhaps of more relevance when we consider demand rather than cost
elasticity. After all, when the seller lacks information about the
willingness-to-pay of the buyer, he may also lack information about the
demand elasticity. Correspondingly, the bound that we obtain is somewhat
weaker as it refers only to the upper bound in the demand elasticity.
Similarly, we do not establish that the uniform price mechanism is Bayes
optimal for arbitrary nonlinear demand functions that satisfy the above
elasticity condition (\ref{eta}).

\section{Procurement}

We focused throughout on the classic problem of nonlinear pricing where the
seller is uncertain about the demand of the buyers who have private
information regarding their willingness-to-pay for varying quality or
quantity. Alternatively, we might be interested in the robust procurement
policies where a single large buyer wishes to procure from sellers who have
private information about their cost condition. We can then ask what are the
robust procurement policies as measured by the competitive ratio. As the
selling and procurement problem are closely connected, we indeed find that a
similar characterization as in Theorem \ref{th:1} and \ref{thm:opti} exists.
As before, a distinction between quantity differentiation and quality
differentiation proves to be useful.

\subsection{Quality Differentiation}

There is a buyer that procures an object with varying quality from a seller
with unknown cost. The buyer has valuation $q\in \mathbb{R}_{+}$ for quality
and the seller has a cost $\theta \cdot c(q)$ to provide a good of quality $%
q $. The parameter $\theta $ is private information for the seller and
described by distribution $F$. The cost function is given by a constant
elasticity 
\begin{equation*}
c(q)=q^{\eta }/\eta .
\end{equation*}%
The efficient social surplus is generated by finding the optimal quality
given the prevailing cost function: 
\begin{equation*}
S(\theta )=\max_{q}\left\{ q-\theta c(q)\right\} .
\end{equation*}%
The efficient quality is inversely related to the cost parameter $\theta :$ 
\begin{equation*}
q^{\ast }=\left( \frac{1}{\theta }\right) ^{\frac{1}{\eta -1}}
\end{equation*}%
and generates a social surplus of:%
\begin{equation*}
S(\theta )=\frac{\eta -1}{\eta }(\frac{1}{\theta })^{\frac{1}{\eta -1}}.
\end{equation*}%
If the buyer offers a constant price $p$ for every marginal unit of quality,
then the seller will optimally offer a quality: 
\begin{equation*}
\theta c^{\prime }(q)=p\Longleftrightarrow q=(\frac{p}{\theta })^{\frac{1}{%
\eta -1}}\text{.}
\end{equation*}

\begin{corollary}[Surplus Guarantee Mechanism]
\label{c:sgm}\quad \newline
The optimal surplus guarantee mechanism offers a constant unit price 
\begin{equation*}
p=1/\eta
\end{equation*}%
for incremental quality and the buyer is guaranteed a share: 
\begin{equation*}
\left( \frac{1}{\eta }\right) ^{\frac{1}{\eta -1}}
\end{equation*}%
of the efficient social surplus.
\end{corollary}

Thus, the robust optimal pricing policy is a uniform unit price for quality
at which the seller can then deliver the optimal quality. The surplus
guarantee is increasing in the elasticity of the cost function of the seller.

If $\theta $ follows a power function $F(\theta )=\theta ^{\alpha }$, the
optimal mechanism consists of maximizing: 
\begin{equation*}
\max \left\{ q-\theta c(q)-c(q)\frac{F(\theta )}{f(\theta )}\right\}.
\end{equation*}%
We get: 
\begin{equation*}
1-\theta (\frac{\alpha +1}{\alpha })c^{\prime }(q)=0.
\end{equation*}%
We thus get: 
\begin{equation*}
q=(\frac{\alpha }{\theta (\alpha +1)})^{\frac{1}{\eta -1}}.
\end{equation*}%
We get the same result when $\alpha =1/(\eta -1)$.

\subsection{Quantity Differentiation}

We can alternatively consider the case where the buyer has a declining
marginal utility for quantity and the seller has a constant marginal cost of
producing additional units of the product. The buyer thus has a utility
function $u(q)$, where $q$ is the quantity of the good and 
\begin{equation}
u(q)=\eta q^{(\eta +1)/\eta }/(\eta +1)  \label{e}
\end{equation}%
with a demand elasticity: 
\begin{equation}
\eta \in \left( -\infty ,-1\right) .  \label{el}
\end{equation}%
The utility function is increasing and concave in the above range with 
\begin{equation*}
u^{\prime }\left( q\right) =q^{\frac{1}{\eta }}>0,u^{\prime \prime }\left(
q\right) =\frac{q^{\frac{1}{\eta }-1}}{\eta }<0.
\end{equation*}%
(For $\eta >-1$, the above parametrization gives a negative gross utility.)
The seller has cost 
\begin{equation*}
c\left( q\right) =\theta \cdot q,
\end{equation*}%
where the marginal cost $\theta $ is private information for the seller and
given by a common prior distribution. The first-best surplus is: 
\begin{equation*}
S(\theta )=\max_{q}\left\{ u(q)-\theta q\right\} .
\end{equation*}%
The efficient quantity is: 
\begin{equation*}
q^{\ast }=\theta ^{\eta }
\end{equation*}%
and the social surplus is%
\begin{equation*}
S(\theta )=-\frac{\theta ^{\eta +1}}{\eta +1}.
\end{equation*}

\begin{corollary}[Surplus Guarantee Mechanism]
\label{c:sgmn}\quad \newline
The optimal surplus guarantee mechanism offers a constant mark-up 
\begin{equation*}
p(q)=\frac{\eta +1}{\eta }q^{1/\eta }
\end{equation*}%
for quantity and the buyer is guaranteed a share: 
\begin{equation*}
\left( \frac{\eta }{\eta +1}\right) ^{\eta +1}
\end{equation*}%
of the optimal social surplus.
\end{corollary}

If the buyer sets a constant mark-up $p(q)=zq^{1/\eta }$, the seller will
sell 
\begin{equation*}
\theta =p(q).
\end{equation*}%
So, $q=(\frac{\theta }{z})^{\eta }$. So the buyer surplus will be: 
\begin{equation*}
u(q)-\int_{0}^{q}p(y)dy=\frac{\eta }{\eta +1}(\frac{\theta }{z})^{\eta +1}-z%
\frac{\eta }{\eta +1}(\frac{\theta }{z})^{\eta +1}
\end{equation*}%
%
%
%
We get that the optimal is: 
\begin{equation*}
z=\frac{\eta +1}{\eta }.
\end{equation*}%
Thus the buyer offers a constant mark-up pricing policy that the price per
unit of quantity delivered is proportional to the marginal valuation.

Thus, for a large buyer, the optimal policy is not a constant unit price but
a constant mark-up price that is decreasing in quantity. (This is easy to
implement and offers a trade-off between commitment and insurance.)

\section{Conclusion\label{sec:con}}

We established robust pricing and menu policies for environments with second
degree price discrimination. We showed that simple pricing policies, namely
constant mark-up policies for the case of quality differentiation and linear
pricing rules for the case of quantity differentiation attain the highest
profit guarantee for the seller.

We established the optimality of these pricing policy for constant
elasticity of cost or demand functions. But we showed that the features of
the policies enable us to establish bounds even beyond the case of constant
elasticity when we merely assume convexity for cost functions or concavity
for demand functions.

In the analysis we focused on the classic optimal selling problem where the
seller is designing a menu of choices to screen the buyers with private
information regarding their willingness to pay. We showed however that the
same arguments and associated mechanism allow us to establish utility
guarantees in procurement settings. Here, the buyer seeks to derive optimal
purchasing policies against vendors with private information regarding their
cost or marginal cost of producing quantities or qualities. Thus, we derived
robust profit or utility guarantees across a wide spectrum of nonlinear
pricing problems.

We formulated the profit guarantee in terms of a competitive ratio relative
to the socially efficient surplus. As part of the minmax problem, the Pareto
distribution emerged as critical distribution that minimizes the revenue of
the seller. As critical value distribution, the Pareto distribution
generates a linear virtual utility. This suggests that a version of the
Pareto distribution may also play an important role in related problems. For
example, \cite{bebm15} consider the limits of third-degree price
discrimination where the consumers all have unit-demand. It is an open
problem how market segmentation can impact the surplus distribution in the
context of nonlinear pricing problems, thus when we allow jointly for second
and third-degree price discrimination. In related work, \cite{behm23} we
investigate the consumer surplus maximizing distribution for nonlinear
pricing problems in the spirit of \cite{rosz17} and \cite{dero22} who
consider single and many-item allocation problems respectively.\newpage

\newpage

\bibliographystyle{econometrica}
\bibliography{general}

\end{document}